\documentclass[ reprint, aps, pra, superscriptaddress, longbibliography, nofootinbib ]{revtex4-2}
\usepackage{amsmath,amsthm,amssymb}

\usepackage{xcolor}
\usepackage[citecolor=blue,linkcolor=blue,urlcolor=blue,colorlinks=true]{hyperref}
\usepackage{physics}
\usepackage{bm}
\usepackage{tikz}
\usepackage{algpseudocode}
\usepackage{lipsum}

\usepackage{multirow}
\usepackage{times}

\newtheorem{theorem}{Theorem}
\newtheorem{lemma}{Lemma}
\newtheorem{example}{Example}
\newtheorem{corollary}{Corollary}

\newcounter{algorithm}
\renewcommand{\thealgorithm}{\arabic{algorithm}}

\makeatletter
\renewcommand*\env@matrix[1][*\c@MaxMatrixCols c]{%
\hskip -\arraycolsep \let\@ifnextchar\new@ifnextchar \array{#1}}
\makeatother
\begin{document}
\setcounter{MaxMatrixCols}{20}

\makeatletter
\def\selectlanguage#1{}
\makeatother

\algrenewcommand\algorithmicrequire{{\bf{Input:}}}
\algrenewcommand\algorithmicensure{{\bf{Output:}}}

\title{Magic Gate Teleportation: Structure, Useful Resource States, and Simpler Feedforward}

\author{Yunzhe Zheng}
\email{yunzhe.zheng@yale.edu}
\affiliation{Yale Quantum Institute \& Department of Applied Physics, Yale University, New Haven, CT, USA}
\author{Allen Zang}
\affiliation{Pritzker School of Molecular Engineering, University of Chicago, Chicago, IL, USA}
\author{Aleksander Kubica}
\email{a.kubica@yale.edu}
\affiliation{Yale Quantum Institute \& Department of Applied Physics, Yale University, New Haven, CT, USA}

\begin{abstract}
Quantum gate teleportation is a key technique in fault-tolerant quantum computation that uses resource states to implement logical gates. Here, we develop a theory of magic gate teleportation (MGT), a form of quantum gate teleportation that uses non-stabilizer resource states to implement non-Clifford unitary gates on arbitrary input states without revealing information about those states. This setting is motivated by the canonical approach to implementing logical non-Clifford gates in fault-tolerant quantum computation. We uncover a hidden structure within MGT---after backpropagating the Pauli measurements, MGT protocols can be viewed as encoding the input state into a stabilizer code heralded by the measurement outcomes, followed by a logical non-Clifford gate. Using this structure, we construct MGT protocols for any resource state obtained by applying commuting Pauli rotations to a stabilizer state, and provide an efficient algorithm for synthesizing their circuit implementations. Conversely, we prove that directly useful single-qubit resource states are necessarily Clifford-equivalent to diagonal states, and establish the same constraint for multi-qubit resource states used in a broad class of MGT protocols. In particular, the output state distilled from the $[\![5, 1, 3]\!]$ protocol is not directly useful for any MGT protocols, though it can be probabilistically converted into directly useful states. Finally, we identify conditions under which the feedforward operators can be implemented by Pauli operators, shedding light on the paradigm of algorithmic fault tolerance and simplifying the feedforward operations needed for quantum computing.
\end{abstract}
\maketitle

\section{Introduction}

Quantum computers are required to be implemented with fault-tolerant logical operations~\cite{Shor1996, Preskill1998} in order to reliably execute quantum algorithms~\cite{nielsenQuantumComputationQuantum2010,Dalzell2025}. One of the simplest ways to implement such operations is via transversal gates. Many quantum error-correcting (QEC) codes~\cite{Shor1995,Steane1996,Gottesman1996} admit transversal Clifford gates; however, universality requires the ability to implement fault-tolerant non-Clifford gates. Unfortunately, there is a no-go theorem~\cite{zeng_transversality_2007, eastin_restrictions_2009} that rules out the possibility of a universal set of transversal gates for any nontrivial QEC codes; similar restrictions also apply to approximate QEC codes~\cite{Faist2020,Kubica2021}. Moreover, non-Clifford gates are typically harder to realize, justifying the common assumption that, at the logical level, Clifford operations are straightforward to implement and the main resource cost lies in the implementation of non-Clifford operations~\cite{aaronson_improved_2004,fowler_surface_2012,Campbell2017}.

A common approach to universality is to implement non-Clifford gates via quantum gate teleportation~\cite{zhou_methodology_2000,gottesman_demonstrating_1999} using resource states that are often referred to as magic states. The paradigmatic example is the $T$ gate teleportation~\cite{zhou_methodology_2000, bravyi_universal_2005}, which uses the $\ket T$ state and, together with Clifford gates, enables universal quantum computation. More generally, any unitary can be (approximately) expressed using Clifford and $T$ gates, and then compiled into the preparation of $\ket 0$ and $\ket T$ states, followed by multi-qubit Pauli measurements (possibly conditioned on previous measurement outcomes)~\cite{litinski_game_2019}. This, in fact, is one of the canonical approaches to fault-tolerant computation with quantum low-density parity-check codes~\cite{horsman_surface_2012,cohen_low-overhead_2022,he_extractors_2025}. Recent works have also explored scalable magic state injection into qLDPC codes~\cite{zhu_topological_2025,zhang_constant_overhead_2025,liu_in_situ_2026}.

The $\ket T$ state is only one example of a resource state for a non-Clifford gate; other resource states can also be useful for implementing different gates via quantum gate teleportation. From the perspective of resource theories, non-stabilizer states should be viewed as computational resources~\cite{veitchResourceTheoryStabilizer2014,howard_application_2017,heinrich_robustness_2019,beverland_lower_2020,wangEfficientlyComputableBounds2020,leone_stabilizer_2022}, and non-stabilizer states other than $\ket T$ states may still give rise to useful non-Clifford gates. However, it remains unclear how one can start with generic non-stabilizer states and consume them for non-Clifford gates through quantum gate teleportation. In particular, one may hope to reduce the overhead of circuit synthesis by compiling quantum algorithms not only into Clifford and $T$ gates~\cite{dawson_solovay_2006,kliuchnikov_asymptotically_2013,selinger_efficient_2014,kliuchnikov_practical_2016,ross_optimal_2016}, but also into Pauli rotations made available by more exotic magic states~\cite{landahlComplexInstructionSet2013,howard_small_2016,choi_fault_2023,akahoshi_partially_2024,yoshioka_transversal_2025,ismail_star_2025,zheng_magic_2024}. The general theory of quantum gate teleportation and useful resource states remains underexplored, with limited understanding of the essential structure of gate teleportation protocols and the usefulness of different resource states. Exploring this research direction is timely, given its importance for quantum computation with quantum low-density parity-check codes.

\begin{figure*}[ht!]
\centering
\includegraphics[width=0.93\linewidth]{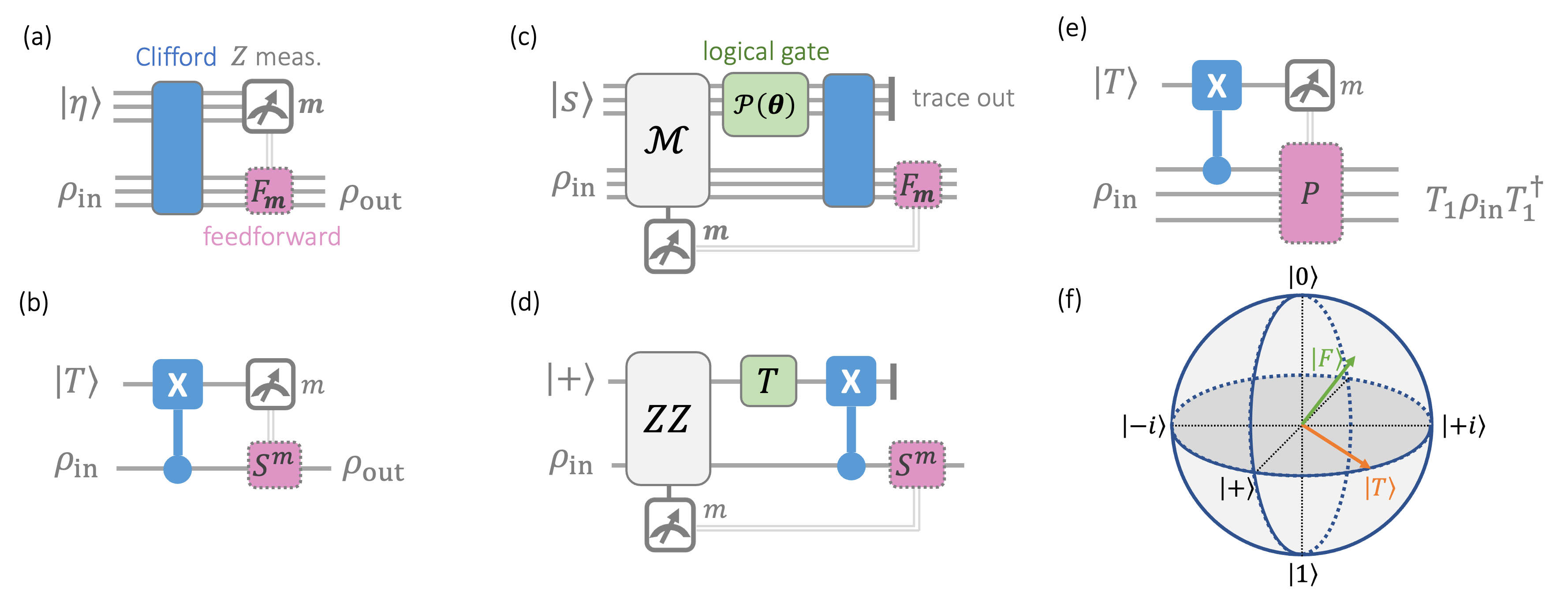}
\caption{ (a) Magic gate teleportation (MGT) is a type of quantum gate teleportation that uses a non-stabilizer resource state $\ket\eta$ to implement a non-Clifford gate on an arbitrary input state $\rho_{\mathrm{in}}$ without revealing any information about it. MGT comprises a Clifford circuit, followed by Pauli $Z$ measurements of all the qubits in the resource register; applying feedforward operators $F_{\bm m}$ for the measurement outcome $\bm m$ results in a deterministic MGT protocol. (b) The canonical MGT protocol that uses $\ket{T}=T\ket{+}$ and implements the $T$ gate. (c) Any MGT protocol using $\ket{\eta}=\mathcal{P}(\bm \theta)\ket{s}$, where $\mathcal{P}(\bm \theta)$ commutes with the backpropagated measurements $\mathcal{M}$, can be viewed as first encoding $\rho_{\mathrm{in}}$ into a stabilizer code heralded by the measurement outcome $\bm m$, followed by implementing a logical gate via $\mathcal P(\bm\theta) \otimes I$. (d) For the MGT protocol in (b), the backpropagated measurement $Z\otimes Z$ encodes $\rho_{\mathrm{in}}$ into a two-qubit repetition code and $T \otimes I$ implements a logical gate. (e) When the input state $\rho_{\mathrm{in}}$ is stabilized by certain Pauli operators, the feedforward operators in MGT can be replaced by Pauli operators. This observation sheds light on the paradigm of algorithmic fault tolerance~\cite{zhou_algorithmic_2024} and allows us to reduce the number of feedforward Clifford operators in quantum computing. (f) Not every non-stabilizer state is directly useful for MGT. For instance, all directly useful single-qubit resource states lie on the three great circles (blue) on the Bloch sphere and the state $\ketbra{F}{F} = \frac{1}{2}I+\frac{1}{2\sqrt{3}}(X+Y+Z)$ from the five-qubit distillation protocol~\cite{bravyi_universal_2005} cannot be used directly as a single-qubit MGT resource. }
\label{fig:Structural}
\end{figure*}

In this article, we study magic gate teleportation (MGT), a type of quantum gate teleportation that uses non-stabilizer resource states and implements non-Clifford gates on arbitrary input states without revealing information about those states. This setting is motivated by the canonical approach to fault-tolerant quantum computation, in which logical Clifford operations can be implemented directly with relatively low overhead~\cite{fowler_surface_2012,Campbell2017}, while logical non-Clifford gates are implemented by consuming magic states through gate teleportation~\cite{gottesman_demonstrating_1999,bravyi_universal_2005,litinski_game_2019}. In this setting, a logical gate must implement the intended unitary for any input state. We also require MGT protocols to preserve the quantum information in that input state because the state may be the result of a long and costly computation, and losing it could require restarting the whole computation process. We call a resource state directly useful for MGT if it can be consumed by an MGT protocol without first being converted into a different resource state. This focus is motivated by fault-tolerant architectures in which magic state factories supply directly useful resource states to the computational circuit~\cite{ogorman_realistic_2017,gidney_efficient_2019,litinski_magic_2019}. Any conversion needed to obtain such states can be regarded as part of the factory, before the resource states interact with the computational input.

We show how non-stabilizer resource states obtained by applying commuting Pauli rotations to stabilizer states can be used to implement non-Clifford gates via MGT. Conversely, we show that if resource states can be used directly in MGT, then they must have a specific form; in particular, for single-qubit resource states, we prove that directly useful states must be Clifford-equivalent to $\frac{1}{\sqrt{2}}(\ket{0}+e^{i\theta}\ket{1})$ for an arbitrary angle $\theta$. We also generalize our finding to the multi-qubit case and prove that directly useful states in a broad class of MGT protocols must be Clifford-equivalent to diagonal states in the computational basis. Our result highlights the subtlety that not every non-stabilizer state is directly useful for MGT, even though such a state may be converted into a useful resource for quantum computation. We also provide an algorithm to construct explicit circuit implementations of magic gate teleportation protocols. Finally, we investigate conditions under which Clifford or non-Clifford feedforward operators in MGT protocols can be replaced by Pauli operators, reducing the practical overhead associated with their implementation. Our results shed light on the paradigm of algorithmic fault tolerance~\cite{zhou_algorithmic_2024,cain2025fast,serra-peralta_decoding_2026}, providing an alternative explanation for why certain logical measurements that control feedforward logical Clifford operators can be replaced by ``coin tosses'' and accounted for with the Pauli frame update.

\section{Preliminaries}

In this section, we briefly review the stabilizer formalism~\cite{gottesman_stabilizer_1997}. Then, we discuss magic gate teleportation, which is a method of implementing non-Clifford gates on an arbitrary input state without revealing any information about it.

\subsection{Stabilizer formalism and the Clifford hierarchy}

For an $n$-qubit pure state $\ket{\psi}$, we say that an $n$-qubit Pauli operator $P$ stabilizes $\ket{\psi}$ if $P\ket{\psi} = \ket{\psi}$. We then define the stabilizer group $\mathcal S$ of $\ket\psi$ as the group generated by all Pauli operators that stabilize $\ket\psi$. Following Ref.~\cite{beverland_lower_2020}, we can define the stabilizer nullity $\nu(\ket{\psi}) = n - r$, where $r$ is the number of independent generators of $\mathcal S$ and refer to states with stabilizer nullity equal to zero as stabilizer states.

We also use stabilizer groups in the context of QEC. Namely, for an $n$-qubit system, the codespace of a stabilizer code $\mathcal Q$ is the $+1$ eigenspace of all elements of an Abelian subgroup $\mathcal{S}$ of the $n$-qubit Pauli group with $-I\not\in\mathcal S$. When we discuss a stabilizer code $\mathcal Q$, we not only refer to its codespace, but also to a specific choice of the logical Pauli operators.

It will be convenient to introduce the following notation. Let $\mathcal{P}=\{P_1,P_2,\ldots\}$ and $\bm{\theta}=\{\theta_1,\theta_2,\ldots\}$ be ordered lists of $n$-qubit Pauli operators and angles, respectively. We define
\begin{equation}
\mathcal{P}(\bm{\theta})= \prod_{j} P_j (\theta_j)= \ldots P_2(\theta_2)P_1(\theta_1),
\end{equation}
where $P_j(\theta_j)=\exp(-i\theta_j P_j/2)$ is an $n$-qubit Pauli rotation by angle $\theta_j$. Up to a global phase, any unitary can be decomposed as a product of Pauli rotations $\mathcal{P}(\bm{\theta})$~\cite{nielsenQuantumComputationQuantum2010,hegde_pauli_2016}.

We will also use the notion of the Clifford hierarchy \cite{gottesman_demonstrating_1999,zeng_transversality_2007,zeng_semiclifford_2008,anderson_classification_2016}. The first level of the Clifford hierarchy $\mathcal C_1$ is defined as the $n$-qubit Pauli group, and, for $k\geq 2$, the $k$th level is defined recursively as
\begin{equation}
    \mathcal C_k
    =
    \{\text{$n$-qubit unitary $U$} | \forall P\in\mathcal C_1: UPU^\dagger \in \mathcal C_{k-1} \}
\end{equation}
up to a global phase~\cite{zeng_transversality_2007}.

The characterization of diagonal gates in the Clifford hierarchy~\cite{cui_diagonal_2017,anderson_groups_2024} implies that, if a diagonal gate $U\in\mathcal C_k$, then it admits a representation $U= \mathcal{P}(\bm \theta)$ with a set of commuting Pauli operators $\mathcal{P}$ and  $\theta_j \in (\frac{\pi}{2^{k-1}})\mathbb Z$ for every $\theta_j \in \bm \theta$. For example, any gate in $\mathcal{C}_3$ of this form can be expressed using Pauli-rotation angles $\theta_j\in(\frac{\pi}{4})\mathbb Z$. This includes familiar examples such as $T=Z(\frac{\pi}{4})$ and $CCZ=\mathcal P(\bm\theta)$, where $\mathcal P=\{Z_1,Z_2,Z_3,Z_1Z_2,Z_1Z_3,Z_2Z_3,Z_1Z_2Z_3\}$ and $\bm\theta=(\frac{\pi}{4},\frac{\pi}{4},\frac{\pi}{4},-\frac{\pi}{4},-\frac{\pi}{4},-\frac{\pi}{4},\frac{\pi}{4})$ up to a global phase.

\subsection{Magic gate teleportation}

Quantum gate teleportation~\cite{gottesman_demonstrating_1999,zhou_methodology_2000} is a method that implements a target gate $U$ on an unknown arbitrary input state $\rho_{\mathrm{in}}$. To achieve that, it consumes some resource state $\ket{\eta}$ (which we assume to be a pure state) and uses operations from a specified set of free operations. We focus on quantum gate teleportation for non-Clifford gates with the following free operations: state preparation in the computational basis, Clifford gates and single-qubit Pauli $Z$ measurements; see Fig.~\ref{fig:Structural}(a).

Importantly, throughout this article, we require that the input state is always recoverable from the output of quantum gate teleportation, i.e., no information about $\rho_{\mathrm{in}}$ is revealed throughout the protocol. This requirement is motivated by the fact that quantum gate teleportation is used to implement logical gates during fault-tolerant quantum computation and the encoded state must not be revealed in the middle of computation. Concretely, we consider protocols which, for an arbitrary input state $\rho_{\mathrm{in}}$, return the output state
\begin{equation}
\rho_{\mathrm{out}}(\bm m ) = U_{\bm m}\rho_{\mathrm{in}}U_{\bm m}^{\dagger},
\end{equation}
where $U_{\bm m}$ is a unitary, conditioned on the outcome $\bm m \in \{0,1\}^*$ of single-qubit Pauli $Z$ measurements. Since we are interested in implementing non-Clifford gates, at least one of the unitaries in $\{U_{\bm m}\}$ should be non-Clifford; we assume it is $U_{\bm 0}$. For brevity, we refer to such protocols as magic gate teleportation (MGT). We then say that an MGT protocol uses the resource state $\ket{\eta}$ and implements $U_{\bm m}$ heralded by the outcome $\bm m$. Also, we can modify any MGT protocol to make it deterministically implement a non-Clifford unitary $U_{\bm 0}$ by applying the following feedforward operators
\begin{equation}
    F_{\bm m} =  U_{\bm 0} U^\dag_{\bm m}
    \label{eq:deterministic-feedforward-condition}
\end{equation}
for the outcome $\bm m$. Consequently, its output is $\rho_{\mathrm{out}} = U_{\bm 0}\rho_{\mathrm{in}}U_{\bm 0}^{\dagger}$, regardless of the outcome $\bm m$. We refer to such protocols as deterministic MGT.

While discussing MGT, it is convenient to backpropagate the Pauli $Z$ measurements to the beginning of the protocol; see Fig.~\ref{fig:Structural}(c). Then, MGT may be viewed as measuring a set of independent commuting Pauli operators $\mathcal M = \{ V^\dagger (Z_j\otimes I) V\}$ on the joint state $\ketbra{\eta}{\eta}\otimes\rho_{\mathrm{in}}$ before applying the Clifford circuit $V$, where we use the convention that the first (second) register is the resource (input) state and $Z_j$ is the Pauli $Z$ operator on the $j$th qubit in the resource register. We call $\mathcal{M}$ the backpropagated measurements. Without loss of generality, we assume that every non-identity element of the group $\langle\mathcal M\rangle$ generated by $\mathcal M$ has non-trivial support on both the resource and input registers; otherwise, some product of the backpropagated measurements would effectively measure either the resource state alone (modifying it into some other resource state) or the input state alone (revealing information about it). Also, the assumption that the input state is always recoverable from the output of the MGT protocol is equivalent to the condition that the group of backpropagated measurements $\langle\mathcal M\rangle$ reveals no information about $\rho_{\mathrm{in}}$. Combined with the requirement that MGT should work for an arbitrary input state $\rho_{\mathrm{in}}$, we then obtain
\begin{equation}
\Tr[M_j \ketbra{\eta}{\eta}\otimes\rho_{\mathrm{in}} ] = 0
\label{eq:state_preserving}
\end{equation}
for every non-identity $M_j\in\langle\mathcal M\rangle$.

\section{Structure of MGT protocols}

We now describe the stabilizer-code structure hidden inside MGT protocols. The main observation is that, after backpropagating the final measurements to the beginning, the protocol can be viewed as encoding the input state into a stabilizer code (heralded by the measurement outcomes), followed by applying a non-Clifford logical gate. The following lemma captures such intuition behind a general class of MGT protocols.

\begin{lemma}
Consider an MGT protocol using a resource state $\ket{\eta}=\mathcal{P}(\bm\theta)\ket{s}$, where $\mathcal{P}$ is a set of non-trivial commuting Pauli operators and $\ket{s}$ is a stabilizer state that is not an eigenstate of any non-identity Pauli operator in $\langle\mathcal{P}\rangle$. Let $\mathcal M = \{M_j\}$ and $\bm m = \{ m_j \}$ be the backpropagated measurements and their outcomes. If
\begin{equation}
    [M_j,\mathcal{P}(\bm\theta)\otimes I]=0
\end{equation}
for every $ M_j\in\mathcal M$, then measuring $\mathcal M$ on $\ketbra{\eta}{\eta}\otimes\rho_{\mathrm{in}}$ is equivalent to encoding $\rho_{\mathrm{in}}$ into a stabilizer code $\mathcal Q_{\bm m}$ heralded by $\bm m$ and implementing a logical gate on $\mathcal Q_{\bm m}$ via $\mathcal P(\bm\theta)\otimes I$.
\label{lem:compatible-logical-action}
\end{lemma}

\begin{proof}
Since $\mathcal{P}(\bm\theta)\otimes I$ commutes with every $M_j$, measuring $\mathcal M$ on $\ketbra{\eta}{\eta}\otimes\rho_{\mathrm{in}}$ is equivalent to first measuring $\mathcal M$ on $\ketbra{s}{s}\otimes\rho_{\mathrm{in}}$ and then applying $\mathcal{P}(\bm\theta)\otimes I$. The codespace projector for the measurement outcome $\bm m$ is
\begin{equation}
    \Pi_{\bm m}
    \propto
    \prod_j (I+(-1)^{m_j}M_j).
\end{equation}
In addition, no information about the input state is revealed based on Eq.~\eqref{eq:state_preserving}. Thus, the measurement encodes $\rho_{\mathrm{in}}$ into the stabilizer code $\mathcal Q_{\bm m}$ whose codespace is defined by $\Pi_{\bm m}$. We also have $[\mathcal{P}(\bm\theta)\otimes I,\Pi_{\bm m}]=0$, so $\mathcal{P}(\bm\theta)\otimes I$ preserves the codespace for each $\Pi_{\bm m}$ and acts as a logical gate.
\end{proof}

To the best of our knowledge, all known MGT protocols~\cite{gottesman_demonstrating_1999,zhou_methodology_2000,bravyi_universal_2005,jones_low-overhead_2013,eastin_distilling_2013,landahlComplexInstructionSet2013} are covered by Lemma~\ref{lem:compatible-logical-action}; we illustrate it with Fig.~\ref{fig:Structural}(c).

\begin{example}
    Consider the deterministic MGT protocol~\cite{bravyi_universal_2005} that uses a resource state $\ket{T}=T\ket{+}$ and implements the $T$ gate; see Fig.~\ref{fig:Structural}(b). The backpropagated measurement is $\mathcal{M}=\{ZZ\}$ with outcome $m$. Lemma~\ref{lem:compatible-logical-action} identifies the stabilizer code as a two-qubit repetition code stabilized by $(-1)^mZZ$ and logical Paulis $\overline X=XX$ and $\overline Z=(-1)^mZI$; the gate $TI$ acts as the logical gate $\overline T\,\overline S^{-m}$. Depending on $m$, the MGT protocol implements either the $T$ or $T^\dagger$ gate and applying feedforward operator $S^m$ makes it a deterministic MGT protocol that implements the $T$ gate; see Fig.~\ref{fig:Structural}(d).
    \label{example:1}
\end{example}

We now use Lemma~\ref{lem:compatible-logical-action} to construct MGT protocols that use resource states represented as Pauli rotations applied to stabilizer states.

\begin{theorem}
Let $    \ket{\eta}=\mathcal{P}(\bm\theta)\ket{s}$ be a resource state, where $\mathcal{P}=\{P_j\}$ is a set of non-trivial commuting Pauli operators and $\ket{s}$ is a stabilizer state that is not an eigenstate of any non-identity Pauli operator in $\langle\mathcal{P}\rangle$. Let $\{P'_k\}$ be a set of independent generators of $\mathcal{P}$ and $\mathcal I_j$ be index sets defined via $P_j=\prod_{k\in\mathcal I_j}P'_k$. For any such $\{P'_k\}$, there is an MGT protocol that uses $\ket{\eta}$ and implements
\begin{equation}
    U_{\bm m} = \prod_j P_j\left((-1)^{q_j}\theta_j\right)
    \label{eq: heralded_gate}
\end{equation}
heralded by the outcome ${\bm m}=\{ m_k \}$, where $q_j = \sum_{k\in\mathcal I_j}m_k$.
\label{th:1}
\end{theorem}

Since Clifford gates are free operations in MGT, Theorem~\ref{th:1} also gives MGT protocols for any Clifford-equivalent heralded gates $C_1U_{\bm m}C_2$, where $C_1$ and $C_2$ can be arbitrary Clifford operators.

\begin{proof}
We are about to construct an MGT protocol with an input state of the same dimensionality as the resource state. Without loss of generality we assume that the resource state $\ket{\eta}$ satisfies $\nu(\ket{\eta})=n$; otherwise, if $\nu(\ket{\eta})<n$, then $\ket{\eta}$ can be converted using free operations in MGT to a $\nu(\ket{\eta})$-qubit state $\ket{\eta'}$ such that $\nu({\ket{\eta'}})$ is equal to the size of the new state. We choose the backpropagated measurements as $\mathcal{M} = \{P'_k\otimes P'_k \}^{n}_{k=1}$. Based on Lemma~\ref{lem:compatible-logical-action}, therefore, measuring $\mathcal{M}$ will encode $\rho_{\mathrm{in}}$ into a stabilizer code, and $\mathcal{P}(\bm{\theta})\otimes I$ will serve as a logical gate of the stabilizer code.

After measurement of $\mathcal{M}$ with outcome $\bm m = \{ m_k \}$, $\{(-1)^{m_k}P'_k\otimes P'_k \}$ will generate the stabilizer group of the code. All logical Pauli $\overline P_j$ operators of the code are given by $\overline {P_j} = I\otimes P_j$, as they commute with all elements in $\mathcal{M}$. Now we examine the effect of $\mathcal{P}(\bm{\theta})\otimes I$ as a logical gate. Since $P_j=\prod_{k\in\mathcal I_j}P'_k$, the operator $(-1)^{q_j}P_j\otimes P_j$ will be in the stabilizer group. Therefore,
\begin{equation}
\begin{split}
    {\mathcal{P}}(\bm{\theta})\otimes I &= \prod_{j}P_j(\theta_j) \otimes I =  I\otimes \prod_{j}P_j\left((-1)^{q_j}\theta_j\right) \\
    &= {\prod_{j}\overline{ P_j}\left((-1)^{q_j}\theta_j\right)}.
    \label{eq: logical_effect}
\end{split}
\end{equation}
Therefore, the implemented gates are $U_{\bm m}=\prod_j P_j\left((-1)^{q_j}\theta_j\right)$ heralded by $\bm m$, as claimed.
\end{proof}

Although the proof constructs a protocol with input and resource registers of the same size, the construction can extend to larger input states. For an $n$-qubit resource and an arbitrary $N$-qubit input, where $N\ge n$, define $\widetilde P_j=P_j\otimes I^{\otimes(N-n)}$ and similarly $\widetilde P'_k=P'_k\otimes I^{\otimes(N-n)}$ for the independent generators. Choosing the backpropagated measurements $P'_k\otimes\widetilde P'_k$ gives the heralded gates $\widetilde U_{\bm m}=\prod_j\widetilde P_j((-1)^{q_j}\theta_j)$. By permuting the input qubits, the heralded gate can be supported on an arbitrary $n$-qubit subset of the $N$ input qubits.

\begin{corollary}
The MGT protocol in Theorem~\ref{th:1} can be promoted to a deterministic MGT protocol for gate
\begin{equation}
    U_{\bm 0}=\mathcal{P}(\bm \theta)
\end{equation}
by applying the feedforward operator
\begin{equation}
    F_{\bm m}
    =
    \prod_j P_j^{q_j}(2\theta_j)
    \label{eq: feedforward}
\end{equation}
for the outcome ${\bm m}=\{m_k\}$, where $q_j=(\sum_{k\in\mathcal I_j}m_k) \mod 2$. Moreover, if $U_{\bm 0}\in \mathcal{C}_k$, then $F_{\bm m} \in \mathcal{C}_{k-1}$.
\label{cor:deterministic-feedforward}
\end{corollary}

In particular, if we require the feedforward operators to be Clifford, i.e., $F_{\bm m}\in\mathcal C_2$, then this deterministic MGT protocol can only implement gates in $\mathcal C_3$. We remark that the hierarchy statement in Corollary~\ref{cor:deterministic-feedforward} is consistent with the relation between gates to implement and feedforward operators observed in Ref.~\cite{gottesman_demonstrating_1999,zhou_methodology_2000}.

\begin{proof}
It suffices to verify the deterministic-feedforward condition in Eq.~\eqref{eq:deterministic-feedforward-condition} using Eq.~\eqref{eq: heralded_gate}. We have
\begin{equation}
   \begin{split}
    F_{\bm m}U_{\bm m}
    &=
    \prod_j P_j^{q_j}(2\theta_j)
    \prod_j P_j\left((-1)^{q_j}\theta_j\right) \\
    &=
    \prod_j P_j(\theta_j)
    =
    U_{\bm 0},
\end{split} 
\end{equation}
where we used the fact that the $P_j$ operators commute pairwise and that $P_j^{q_j}(2\theta_j)P_j\left((-1)^{q_j}\theta_j\right)=P_j(\theta_j)$. Thus, the feedforward removes the dependence on the outcome $\bm m$ and promotes the MGT protocol to a deterministic one for gate $U_{\bm 0}$.

For the hierarchy statement, since the $P'_k$ operators commute pairwise, we can choose a set of Pauli operators $\{R_k\}$ such that $\{R_k,P'_k\}=0$ and $[R_k,P'_\ell]=0$ for $\ell\neq k$. Let $R_{\bm m}=\prod_k R_k^{m_k}$. Since $P_j=\prod_{k\in\mathcal I_j}P'_k$, we have
\begin{equation}
    R_{\bm m}P_jR_{\bm m}=(-1)^{q_j}P_j,
\end{equation}
and hence $U_{\bm m}=R_{\bm m}U_{\bm 0}R_{\bm m}$. Therefore,
\begin{equation}
    F_{\bm m}=U_{\bm 0}U_{\bm m}^\dagger
    = U_{\bm 0}R_{\bm m}U_{\bm 0}^\dagger R_{\bm m} .
\end{equation}
If $U_{\bm 0}\in\mathcal C_k$, then $U_{\bm 0}R_{\bm m}U_{\bm 0}^\dagger\in\mathcal C_{k-1}$ by the definition of the Clifford hierarchy. Since the level of the Clifford hierarchy is preserved if multiplied by a Pauli operator, it follows that $F_{\bm m}\in\mathcal C_{k-1}$.
\end{proof}

Notably, although the construction here only establishes the abstract existence of such MGT protocols, Sec.~\ref{sec:circuit-construction} will provide an explicit circuit implementation.
\section{Constraints on directly useful resource states}

We now focus on a converse direction to the one in the previous section, i.e., given a resource state that enables an MGT protocol to implement a non-Clifford gate for an arbitrary input state, what are the constraints on that resource state? Here, we call a resource state directly useful for MGT if it can be supplied to the MGT protocol without a separate step that converts it into a different resource state.

\subsection{Single-qubit states}

We first analyze single-qubit resource states, where the condition in Eq.~\eqref{eq:state_preserving} allows us to fully characterize directly useful states.

\begin{theorem}%
Let $\ket{\eta}$ be a single-qubit pure state. If there exists an MGT protocol using $\ket{\eta}$, then
\begin{equation}
    \ket{\eta}=C\cdot Z(\theta)\ket{+},
\end{equation}
where $C$ is a Clifford operator and $Z(\theta)=\exp(-i\theta Z/2)$.
\label{th:single-qubit-teleportable}
\end{theorem}

Theorem~\ref{th:single-qubit-teleportable} shows that not every single-qubit non-stabilizer state is directly useful for implementing non-Clifford gates via MGT protocols. Concretely, the directly useful single-qubit non-stabilizer states lie on the three great circles of the Bloch sphere, as illustrated in Fig.~\ref{fig:Structural}(f). Consequently, the state $\ketbra{F}{F}=\frac{1}{2}(I+(X+Y+Z)/\sqrt{3})$ distilled from the five-qubit magic state distillation protocol~\cite{bravyi_universal_2005} is not directly useful as a single-qubit MGT resource. This contrasts with the stabilizer-polytope viewpoint~\cite{vandamNoiseThresholdsHigher2010,howardQuditVersionsQubit2012,heinrich_robustness_2019}, in which $\ket{F}$ is an extremal single-qubit non-stabilizer state that is farthest from the stabilizer octahedron. To use the $\ket{F}$ state, an additional parity measurement and post-selection may be performed on two $\ket{F}$ states to convert them into a $Z(\pi/3)\ket{+}$ state~\cite{bravyi_universal_2005}. Similarly, two copies of the two-qubit magic state $\ket{CZ}=(\ket{00}+\ket{01}+\ket{10})/\sqrt{3}$ can be probabilistically converted using Pauli measurements and Clifford operations into a Toffoli resource state, which is Clifford-equivalent to $\ket{CCZ}=CCZ\ket{+}^{\otimes 3}$~\cite{gupta_encoding_2024}.

\begin{proof}
Any single-qubit state can be written as $\ket{\eta}=CX(\phi)Z(\theta)\ket{+}$, where $C$ is a Clifford operator and $\phi,\theta$ are rotation angles~\cite{nielsenQuantumComputationQuantum2010}. Since the Clifford $C$ can be absorbed into the Clifford part of the MGT protocol, we take $C=I$ without loss of generality. If $\theta\in(\pi/2)\mathbb Z$, then $Z(\theta)$ is Clifford and can be propagated through $X(\phi)$, converting the latter into a rotation about another Pauli axis. After absorbing the Clifford operations into $C$, the claimed form follows. We may therefore assume that the diagonal component is non-Clifford, i.e., $\theta\notin(\pi/2)\mathbb Z$.

Consider the backpropagated measurement $M=A\otimes B$ appearing in an MGT protocol, where $A$ and $B$ are supported on the resource and input register, respectively. Both $A$ and $B$ have to be nontrivial; otherwise, one will not be implementing an MGT protocol, but rather measuring either the resource or input state. The condition in Eq.~\eqref{eq:state_preserving} requires that
\begin{equation}
    \Tr\!\left[M(\ketbra{\eta}{\eta}\otimes\rho_{\mathrm{in}})\right]
    =
    \bra{\eta}A\ket{\eta}\Tr[B\rho_{\mathrm{in}}] = 0,
    \label{eq_extra}
\end{equation}
for any input state $\rho_{\mathrm{in}}$. Thus, we must have $\bra{\eta}A\ket{\eta}=0$.

We now show by contradiction that an MGT protocol forces $\phi\in(\pi/2)\mathbb Z$. Let us assume that $\phi\notin(\pi/2)\mathbb Z$. For $\ket{\eta}=X(\phi)Z(\theta)\ket{+}$, the expectation values for all possible non-trivial $A\in\{X,Y,Z\}$ are given by
\begin{equation}
\begin{split}
    \bra{\eta}X\ket{\eta}
    &=
    \cos\theta,\\
    \bra{\eta}Y\ket{\eta}
    &=
    \cos\phi\sin\theta,\\
    \bra{\eta}Z\ket{\eta}
    &=
    \sin\phi\sin\theta .
\end{split}
\end{equation}
Since both $\phi,\theta\notin(\pi/2)\mathbb Z$, none of these expectation values will be zero. Therefore, no choice of nontrivial Pauli $A\in\{X,Y,Z\}$ can satisfy Eq.~\eqref{eq_extra}, contradicting the existence of an MGT measurement that preserves arbitrary input states. Hence, we must have $\phi\in(\pi/2)\mathbb Z$, which means that $X(\phi)$ can be absorbed into the Clifford operator $C$. Therefore, $\ket{\eta} = CZ(\theta)\ket{+}$.
\end{proof}

\subsection{Multi-qubit states}

We would like to have a generalized version of Theorem~\ref{th:single-qubit-teleportable} that is applicable to multi-qubit resource states. To achieve that, the following lemma will be useful.

\begin{lemma}
Let $\ket{\eta}$ be an $n$-qubit pure state. The following statements are equivalent.
\begin{enumerate}
\item There exists a set of $n$ independent commuting Pauli operators $\mathcal{R}=\{R_1,\ldots,R_n\}$ such that
\begin{equation*}
    \bra{\eta}R\ket{\eta}=0,
    \quad
    \forall R\in\langle\mathcal{R}\rangle\setminus\{I\},
\end{equation*}
where $\langle\mathcal{R}\rangle$ denotes the group generated by $\mathcal{R}$.

\item There exists a Clifford operator $C$ and a unitary diagonal in the computational basis $D=\sum_{\bm x\in\{0,1\}^n}e^{i\theta_{\bm x}}\ketbra{\bm x}{\bm x}$, such that $  \ket{\eta}=CD\ket{+}^{\otimes n}$.
\item There exists a set of non-trivial commuting Pauli operators $\mathcal{P}$ and a stabilizer state $\ket{s}$ that is not an eigenstate of any non-identity Pauli operator in $\langle\mathcal{P}\rangle$, such that $\ket{\eta} = \mathcal{P}(\bm \theta)\ket{s}$.
\end{enumerate}
\label{lem:uniform-diagonal-state}
\end{lemma}

\begin{proof}
(1) $\rightarrow$ (2). Since $\mathcal{R}$ is a maximal set of independent commuting Pauli operators, there is a Clifford unitary $C$ such that $C^\dagger R_j C=Z_j$ for all $j$. Let $\ket{\phi}=C^\dagger\ket{\eta}$ and write
\begin{equation}
    \ket{\phi}=\sum_{\bm x\in\{0,1\}^n}a_{\bm x}\ket{\bm x}.
\end{equation}
For $\bm s\in\{0,1\}^n$, define $Z(\bm s)=\prod_{j=1}^n Z_j^{s_j}$. The condition in statement 1 implies that, for every nonzero $\bm s$,
\begin{equation}
    \bra{\phi}Z(\bm s)\ket{\phi}
    =\sum_{\bm x\in\{0,1\}^n}(-1)^{\bm s\cdot\bm x}|a_{\bm x}|^2=0.
\end{equation}
Thus, all non-trivial Fourier coefficients of the probability distribution $p_{\bm x}=|a_{\bm x}|^2$ vanish. Together with $\sum_{\bm x}p_{\bm x}=1$, this implies $p_{\bm x}=2^{-n}$ for every $\bm x$. Consequently, $a_{\bm x}=2^{-n/2}e^{i\theta_{\bm x}}$ for some $\theta_{\bm x}$, and we obtain
\begin{equation}
    \ket{\phi}
    =
    \left(\sum_{\bm x}e^{i\theta_{\bm x}}\ketbra{\bm x}{\bm x}\right)
    \ket{+}^{\otimes n}
    =
    D\ket{+}^{\otimes n}.
\end{equation}
Thus, $\ket{\eta}=CD\ket{+}^{\otimes n}$.

(2) $\rightarrow$ (1). If $\ket{\eta}=CD\ket{+}^{\otimes n}$, let $\mathcal R=\{CZ_jC^\dagger\}_{j=1}^n$. Every non-identity element of $\langle\mathcal R\rangle$ can be written as $CZ(\bm s)C^\dagger$ for some nonzero $\bm s$. Since $D$ commutes with $Z(\bm s)$,
\begin{equation}
    \bra{\eta}CZ(\bm s)C^\dagger\ket{\eta}
    =\bra{+}^{\otimes n}Z(\bm s)\ket{+}^{\otimes n}=0.
\end{equation}

(2) $\rightarrow$ (3). Since $D$ is diagonal, we can decompose $D = \prod_j P^{\mathrm{z}}_j(\theta_j)$, where each $P^{\mathrm{z}}_j$ is a multi-qubit Pauli operator that is a tensor product of the identity and Pauli $Z$. Therefore,

\begin{equation}
\ket{\eta} = C D \ket{+}^{\otimes n}
= \prod_j (C P^{\mathrm{z}}_j(\theta_j) C^\dagger) C \ket{+}^{\otimes n}.
\end{equation}
By setting $\mathcal{P} = \{C P^{\mathrm{z}}_j C^\dagger\}$ and $\ket{s} = C \ket{+}^{\otimes n}$, we obtain $\ket{\eta} = \mathcal{P}(\bm \theta)\ket{s}$.

(3) $\rightarrow$ (2). Because $\ket{s}$ is a stabilizer state that is not an eigenstate of any non-identity Pauli operator in $\langle\mathcal{P}\rangle$, there exists a Clifford unitary $C$ that diagonalizes every element in $\mathcal{P}$ and maps $\ket{s}$ to $\ket{+}^{\otimes n}$~\cite{gottesman_stabilizer_1997}. Hence, $\ket{s}=C\ket{+}^{\otimes n}$ and $C^\dagger \mathcal{P}(\bm \theta) C$ is diagonal. Letting $D = C^\dagger \mathcal{P}(\bm \theta) C$, we therefore have $\ket{\eta}= \mathcal{P}(\bm \theta)\ket{s} = C D\ket{+}^{\otimes n}$.
\end{proof}

\begin{theorem}
Let $\ket{\eta}$ be an $n$-qubit pure state. If there exists an MGT protocol using $\ket{\eta}$ with the backpropagated measurements $\mathcal M = \{A_j\otimes B_j\}$, such that all $A_j$ commute with each other, then $\ket{\eta}$ is Clifford-equivalent to a diagonal state, i.e.,
\begin{equation}
    \ket{\eta}=C D\ket{+}^{\otimes n},
\end{equation}
where $C$ is a Clifford operator and $D$ is a unitary diagonal in the computational basis.
\label{th:multiqubit-teleportable}
\end{theorem}

\begin{figure*}[ht!]
    \centering
    \includegraphics[width=.95\linewidth]{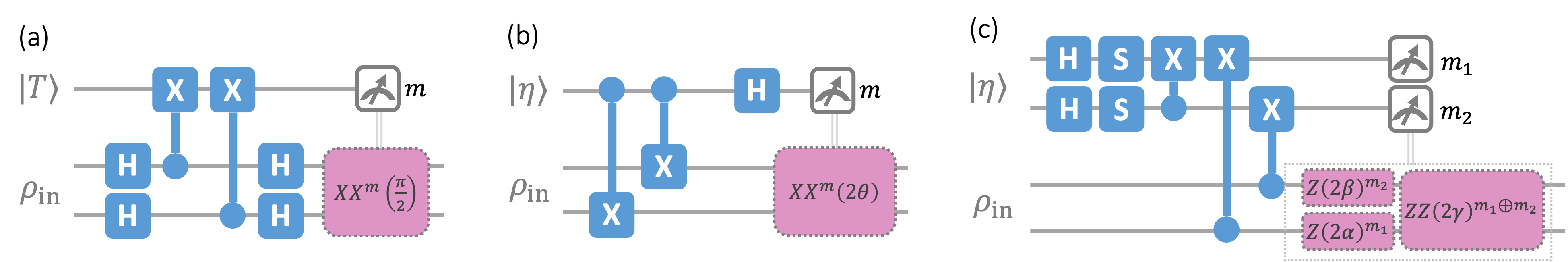}
    \caption{ Circuit implementation of deterministic MGT protocols obtained using Algorithm~\ref{alg:mgt-construction}. Panels (a), (b) and (c) correspond to Examples~\ref{example:0},~\ref{example:a}~and~\ref{example:b}, respectively. }
    \label{fig:implementations}
\end{figure*}

Although Theorem~\ref{th:multiqubit-teleportable} considers a class of MGT protocols satisfying the additional assumption on the commutativity of all $A_j$, this class is still broad enough to include all known MGT protocols. Within this class, any directly useful resource state must be Clifford-equivalent to a diagonal state. It thus highlights a distinction between non-stabilizer states as resources in general resource theories~\cite{veitchResourceTheoryStabilizer2014,howard_application_2017,heinrich_robustness_2019,beverland_lower_2020,wangEfficientlyComputableBounds2020,leone_stabilizer_2022} and their direct usefulness in MGT protocols. Moreover, connected by Lemma~\ref{lem:uniform-diagonal-state} (equivalence between statements 2 and 3), Theorem~\ref{th:multiqubit-teleportable} also shows that any directly useful resource states within this class of MGT protocols must fall into the class of resource states considered in Theorem~\ref{th:1}. Therefore, Theorem~\ref{th:1} also covers a broad scope of MGT protocols.

\begin{proof}

Recall that every non-identity element of $\langle\mathcal M\rangle$ has non-trivial support on both the resource and input registers. Therefore, the operators $A_j$ are independent; otherwise, for some nonempty index set $\mathcal I$, the product $\prod_{j\in\mathcal I}A_j=I$, and hence $\prod_{j\in\mathcal I}(A_j\otimes B_j)=I\otimes\prod_{j\in\mathcal I}B_j$ would be a non-identity element of $\langle\mathcal M\rangle$ supported only on the input register.

Let $\mathcal A=\{A_j\}$. For any non-identity $A\in\langle\mathcal A\rangle$, write $A=\prod_{j\in\mathcal I}A_j$ for some nonempty $\mathcal I$, and define $B=\prod_{j\in\mathcal I}B_j$. Applying the state-preserving condition to $A\otimes B\in\langle\mathcal M\rangle$ gives
\begin{equation}
    \bra{\eta}A\ket{\eta}\Tr[B\rho_{\mathrm{in}}]=0
\end{equation}
for every input state $\rho_{\mathrm{in}}$. Since $B$ is a non-identity Pauli operator, one can choose $\rho_{\mathrm{in}}$ such that $\Tr[B\rho_{\mathrm{in}}]\neq0$. Thus, $\bra{\eta}A\ket{\eta}=0$ for every non-identity $A\in\langle\mathcal A\rangle$. Using the equivalence between statements 1 and 2 in Lemma~\ref{lem:uniform-diagonal-state}, we therefore have $\ket{\eta}=CD\ket{+}^{\otimes n}$ for a Clifford operator $C$ and a diagonal unitary $D$.
\end{proof}

\section{Circuit implementation of MGT}
\label{sec:circuit-construction}

We now describe an explicit implementation of MGT. In principle, a deterministic MGT protocol can be implemented by applying a Clifford circuit to both the resource and input register, measuring qubits of the resource register in the $Z$ basis, and applying the corresponding feedforward operator from Corollary~\ref{cor:deterministic-feedforward} to the output qubits. Although we have been treating Clifford operations as free operations, in practice, they still contribute to the physical implementation cost. It is therefore favorable to realize MGT protocols with optimized Clifford circuits, and to implement the neighboring Clifford gates as a change of basis for the teleported non-Clifford gates whenever possible. Indeed, in a circuit comprising interleaved Clifford $C_j$ and non-Clifford $U_j$ gates, each $C_j$ can be propagated forward so that only a final accumulated Clifford gate remains, e.g.,
\begin{equation}
     U_iC_i \ldots  U_1C_1
    =
    \widetilde C_i(\widetilde C^\dagger_i U_i \widetilde C_i) \ldots (\widetilde C^\dagger_1 U_1 \widetilde C_1) .
\end{equation}
where $\widetilde C_j=C_j\ldots C_1$. Thus, except for the final accumulated Clifford gate $\widetilde C_i$, one can implement the circuit by teleporting Clifford-conjugated non-Clifford gates $\widetilde C^\dagger_j U_j \widetilde C_j$ without implementing the intermediate Clifford gates. Motivated by this observation, we provide the following Algorithm~\ref{alg:mgt-construction} that constructs a circuit implementation of MGT protocols.

\begin{center}
\begin{minipage}{0.95\columnwidth}
\hrule
\vspace{0.6ex}
\refstepcounter{algorithm}
\noindent {\bf{Algorithm~\thealgorithm.}} Constructing a circuit implementation of a deterministic MGT protocol
\hrule
\label{alg:mgt-construction}
\vspace{0.4ex}
\begin{algorithmic}[1]
\Require A resource state $\ket{\eta} = \mathcal{P}(\bm \theta) \ket{s}$ that satisfies the condition in Theorem~\ref{th:1}, and a conjugated non-Clifford gate $U_{\bm 0} = C \mathcal{P}(\bm \theta) C^\dagger $ for some Clifford operator $C$.
\Ensure A Clifford circuit $V$; see Fig.~\ref{fig:Structural}(a).
\State Identify a set of non-trivial independent Pauli operators $ \{P'_k\}$ that generate $\mathcal{P}$.
\State Construct the backpropagated measurements $\mathcal{M} = \{P'_k \otimes CP'_kC^\dagger \}$.
\State Obtain the stabilizer code $\mathcal{Q}$ from $\mathcal{M}$ and $\ket{s}$ using the stabilizer tableau update rule.
\State Obtain the decoding circuit $V$ for $\mathcal{Q}$.
\end{algorithmic}
\vspace{0.6ex}
\hrule
\end{minipage}
\end{center}

Let us now discuss Algorithm~\ref{alg:mgt-construction}. The first two steps are obvious. In Step 3, we interpret the state $\ketbra{s}{s}\otimes\rho_{\mathrm{in}}$ as a code stabilized by $\mathcal{S}_{\mathrm{s}}\otimes I$, where $\mathcal{S}_{\mathrm{s}}$ is the stabilizer group of $\ket{s}$. The logical operators can be chosen as $\overline X_j = I\otimes X_j$ and $\overline Z_j = I\otimes Z_j$. Let the measurement outcomes of $\mathcal{M}$ be $\bm{m} = \{m_k \}$. Defining $\mathcal{M}(\bm m) = \{(-1)^{m_k} M_k\}$, the new stabilizer code $\mathcal{Q}_{\bm m}$ after measuring $\mathcal{M}$ is stabilized by $\langle \mathcal{C}(\mathcal{M}), \mathcal{M}(\bm m) \rangle$, where $\mathcal{C}(\mathcal{M})$ is the set composed of all the elements in $\mathcal{S}_{\mathrm{s}}\otimes I$ that commute with $\mathcal M$. The logical operators $\{ \overline L\}$ may need to be updated after each measurement to ensure that they commute with the new stabilizer group. Concretely, when one measures $\tilde M \in \mathcal{M}$, if $\{L, \tilde M\} = 0$, then one selects an arbitrary element $g$ in the stabilizer group before measuring $\tilde M$ such that $\{g, \tilde M\} =0$ and updates $L' = g L$  so that $[L', \tilde M]=0$; otherwise, there is no need to update $L$. Such an update can be efficiently determined using the tableau simulation~\cite{aaronson_improved_2004}. In Step 4, we synthesize a Clifford decoding circuit $W_{\mathrm{dec}}$ that decodes logical states in $\mathcal{Q}_{\bm m }$ to input register and stabilizer generators to Pauli $Z$s in ancilla register up to a sign using the algorithm in Ref.~\cite{gottesman_stabilizer_1997}. Since $\mathcal{Q}_{\bm m }$ for different $\bm m$ only differ in the sign of stabilizers, they share the same circuit $V$, and we discuss the details in Appendix \ref{app:gottemans_algorithm}. Once we have the Clifford circuit, we measure all qubits of the resource register in the $Z$ basis. Finally, the feedforward operator in Eq.~\eqref{eq: feedforward} should be applied on the decoded state based on the $\bm m$.

Below we provide some concrete examples of deterministic MGT protocols obtained using Algorithm~\ref{alg:mgt-construction}.

\begin{example}
The deterministic MGT protocol that uses $\ket{T}$ and implements the $XX(\frac{\pi}{4})$ gate is obtained by measuring $\mathcal{M}=\{Z_1X_2X_3\}$. The corresponding stabilizer code $\mathcal{Q}$ is a $[\![3, 2]\!]$ code stabilized by $\mathcal{S} = \langle (-1)^m Z_1X_2X_3\rangle$, with logical operators $\overline X_j = X_{j+1}$ and $\overline Z_j = X_1Z_{j+1}$ for $j=1,2$. The feedforward operator is $XX(\pi/2)$ for outcome $m=1$; see Fig.~\ref{fig:implementations}(a).
\label{example:0}
\end{example}

\begin{example}
The deterministic MGT protocol that uses the single-qubit resource state $\ket{\eta}=X(\theta)\ket{0}$ and implements the $XX(\theta)$ gate is obtained by measuring $\mathcal{M}=\{X_1X_2X_3\}$. The corresponding stabilizer code $\mathcal{Q}$ is a $[\![3,2]\!]$ code stabilized by $\mathcal{S}=\langle(-1)^mX_1X_2X_3\rangle$, with logical operators $\overline X_j=X_{j+1}$ and $\overline Z_j=Z_1Z_{j+1}$ for $j=1,2$. On the codespace, $X_1=(-1)^m\overline X_1\overline X_2$, so the resource rotation $X_1(\theta)$ implements the heralded logical gate $\overline X_1\overline X_2((-1)^m\theta)$. The feedforward operator is $XX(2\theta)$ for outcome $m=1$; see Fig.~\ref{fig:implementations}(b).
\label{example:a}
\end{example}

\begin{example}
The deterministic MGT protocol that uses $\ket{\eta} = X_1X_2(\alpha)X_1(\gamma)X_2(\beta)\ket{+i}^{\otimes 2}$ and implements the $Z_1(\beta)Z_2(\alpha)Z_1Z_2(\gamma)$ gate is obtained by measuring $\mathcal{M} = \{X_1X_2Z_4, X_2Z_3\}$. Here, $m_1$ and $m_2$ denote the measurement outcomes on the upper and lower resource qubits, respectively. The corresponding stabilizer code $\mathcal{Q}$ is a $[\![4, 2]\!]$ code stabilized by $\mathcal{S} = \langle (-1)^{m_1}X_1X_2Z_4, (-1)^{m_2}X_2Z_3\rangle$, with logical operators $\overline X_1 = Y_1Y_2X_3$, $\overline X_2 = Y_1X_4$, and $\overline Z_j = Z_{j+2}$ for $j=1,2$. The feedforward operator is $Z_1(2\beta)^{m_2} Z_2(2\alpha)^{m_1} Z_1Z_2(2\gamma)^{m_1\oplus m_2}$; see Fig.~\ref{fig:implementations}(c).
\label{example:b}
\end{example}

\section{When feedforward is ``simple''}

We now discuss scenarios when the feedforward operators in MGT are ``simple''. Concretely, we identify conditions that guarantee Pauli feedforward operators; we also discuss the case of Clifford feedforward operators. Our findings shed light on the paradigm of algorithmic fault tolerance~\cite{zhou_algorithmic_2024} and allow us to reduce the number of feedforward Clifford operators in quantum computing.

\subsection{Criterion for Pauli feedforward}
\begin{theorem}
Consider an MGT protocol described in Theorem~\ref{th:1}. Suppose that the input state $\rho_{\mathrm{in}}$ is stabilized by a set of commuting Pauli operators $\{G_k\}$, i.e., $G_k\rho_{\mathrm{in}}G_k^\dagger=\rho_{\mathrm{in}}$, such that for all $k$ and $\ell\neq k$
\begin{equation}
\{G_k,P'_k\}=0
\quad\text{and}\quad
[G_k,P'_\ell]=0.
\label{eq_1}
\end{equation}
Then, this MGT protocol can be promoted to a deterministic MGT protocol that implements $U_{\bm 0}=\mathcal{P}(\bm\theta)$ with Pauli feedforward operators
\begin{equation}
F_{\bm m}=\prod_k G_k^{m_k}
\end{equation}
for the outcome ${\bm m} = \{ m_k \}$.
\label{th:feedforward_reduction}
\end{theorem}

\begin{proof}
From Eq.~\eqref{eq: logical_effect}, the heralded output state before the feedforward operator is $\rho_{\mathrm{out}}(\bm m)=U_{\bm m}\rho_{\mathrm{in}}U_{\bm m}^\dagger$, where $U_{\bm m} = \prod_j P_j((-1)^{q_j}\theta_j)$. Since $P_j=\prod_{k\in\mathcal I_j}P'_k$, the commutation conditions in Eq.~\eqref{eq_1} imply
\begin{equation}
    F_{\bm m}P_jF_{\bm m}^\dagger=(-1)^{q_j}P_j .
\end{equation}
Therefore, $F_{\bm m } P_j((-1)^{q_j}\theta_j) = P_j(\theta_j) F_{\bm m }$ for every $j$, and we have

\begin{equation}
    F_{\bm m } U_{\bm m } = U_{\bm 0} F_{\bm m }.
\end{equation}
Applying $F_{\bm m}$ to $ \rho_{\mathrm{out}}(\bm m)$ yields
\begin{equation}
    F_{\bm m} \rho_{\mathrm{out}}(\bm m) F_{\bm m}^\dagger
    = U_{\bm 0} F_{\bm m} \rho_{\mathrm{in}} F_{\bm m}^\dagger U^\dagger_{\bm 0}
    = U_{\bm 0} \rho_{\mathrm{in}} U^\dagger_{\bm 0},
\end{equation}
where we used the fact that every $F_{\bm m}$ also stabilizes $\rho_{\mathrm{in}}$. Therefore, the feedforward operators $\{F_{\bm m}\}$ result in a deterministic MGT protocol for the specified input state $\rho_{\mathrm{in}}$.
\end{proof}

Note that if $\rho_{\mathrm{in}}$ is a pure stabilizer state and no nontrivial product of the independent generators $\{P'_j\}$ belongs to the stabilizer group of $\rho_{\mathrm{in}}$ up to a phase, then one can always choose stabilizers $\{G_k\}$ of $\rho_{\mathrm{in}}$ satisfying Eq.~\eqref{eq_1}. Theorem~\ref{th:feedforward_reduction} therefore guarantees that such MGT protocols can be promoted to deterministic MGT with Pauli feedforward for this input state. In addition, we note that the auto-feedforward method~\cite{gidney_flexible_2019} also reduces Clifford feedforward latency by preparing resource states for Clifford feedforward operators in advance, but increases the average resource cost because these states are consumed even when the corresponding feedforward operators are unnecessary. Our criterion instead identifies input states for which Clifford feedforward operators can be replaced by Pauli operators tracked in the Pauli frame without additional resources. Therefore, our result is distinct from the previous auto-feedforward method.

\begin{example}
    Consider a deterministic MGT protocol that uses $\ket{T}$ and implements the $T$ gate. For the input state $\rho_{\mathrm{in}}=\ketbra{+}{+}$, the feedforward operator $S^m$ in Example~\ref{example:1} can be replaced by a Pauli operator $X^m$ . This is because $X$ stabilizes $\ket{+}$ and $\{X,Z\}=0$.
\end{example}

\begin{example}
    Consider a deterministic MGT protocol that uses $Z_1(\beta) Z_2(\theta)\ket{+}^{\otimes 2}$ and implements the $Z_1(\theta) Z_2(\beta)$ gate. For the input state $\ket{s}$ stabilized by $\langle X_1Z_2, Z_1X_2 \rangle$, the feedforward operators can be replaced by $(X_1Z_2)^{m_2} (Z_1X_2)^{m_1}$; see Fig.~\ref{fig:tft_pauli_feedforward}(c).
\end{example}

\subsection{Connection to algorithmic fault tolerance}

\begin{figure}[t]
    \centering
    \includegraphics[width=.55\columnwidth]{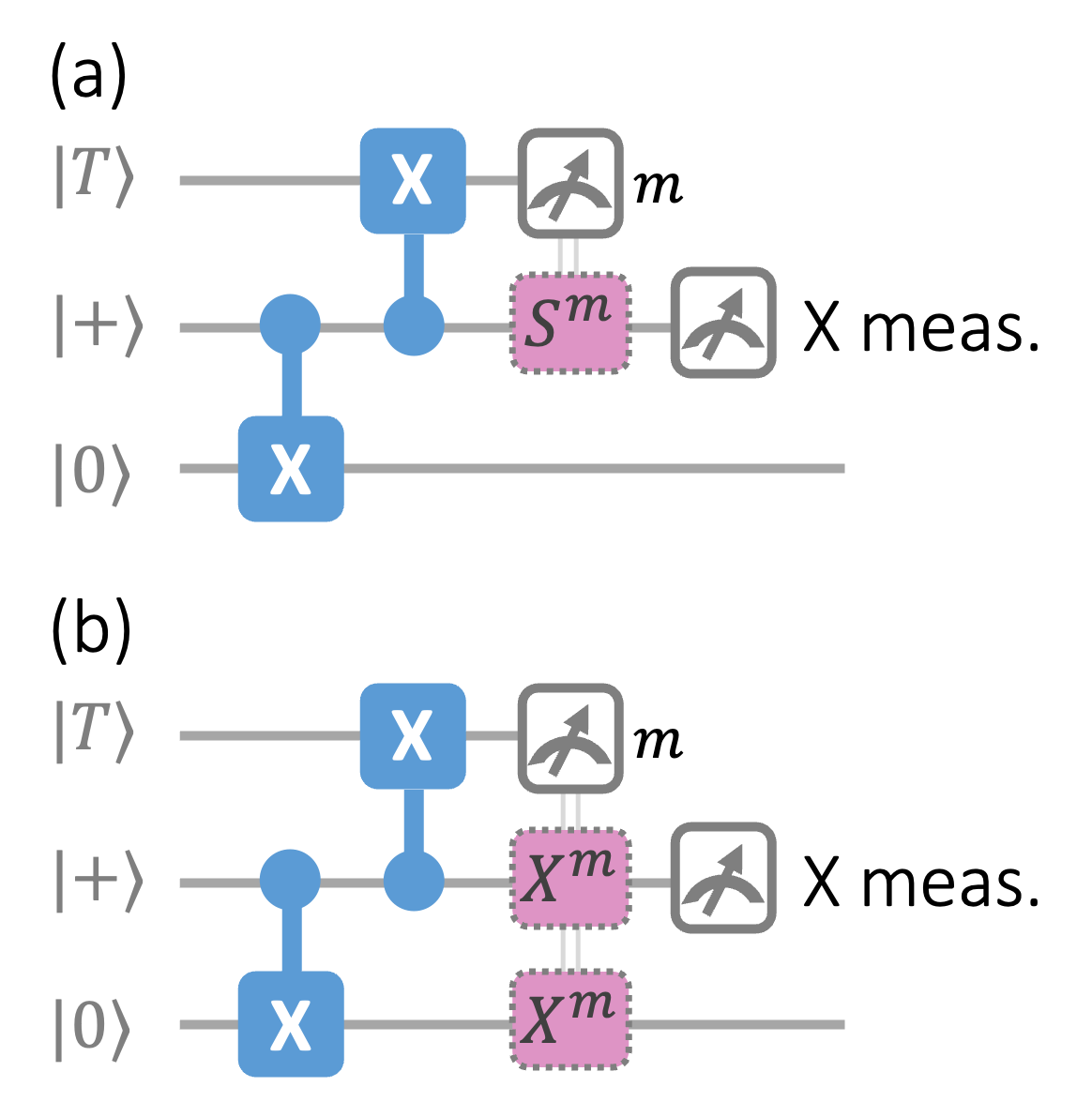}
    \caption{(a) An example of the $T$ gate teleportation followed by the $X$ measurement, where $m$ is the teleportation measurement outcome. Incorrectly inferring $m$ changes the second measurement basis from $X$ to $Y$. (b) The same protocol with $S^m$ replaced by Pauli feedforward $(X\otimes X)^m$ which stabilizes the input state. Incorrectly inferring $m$ produces only a Pauli discrepancy, which can be accounted for by classical postprocessing.}
    \label{fig:MGT_AFT}
\end{figure}

The paradigm of algorithmic fault tolerance~\cite{zhou_algorithmic_2024,cain2025fast,serra-peralta_decoding_2026} asserts that logical circuits composed of transversal Clifford operations supplemented with $\ket T$ states to implement non-Clifford $T$ gates can be executed fault-tolerantly without performing repeated rounds of syndrome extraction between every pair of consecutive operations. Remarkably, this paradigm applies to a broad class of QEC codes, including the surface code~\cite{Kitaev2003,Dennis2002}, rather than being limited to QEC codes with the single-shot QEC property~\cite{Bombin2015,Campbell2019,Kubica2022,Gu2024}. As a result, it can lead to an order-of-magnitude reduction in the space and time overhead of fault-tolerant quantum computation.

The crux of algorithmic fault tolerance is that logical measurement outcomes fall into two types: they can either be reliably inferred using correlated decoding of the syndrome history, or they are 50/50 random and independent of all previous outcomes. In the latter case, the initial inferred value of the logical measurement outcome can be assigned by an unbiased coin flip. This inference may later be revised as additional syndrome information becomes available from syndrome measurements interspersed with the logical circuit. At first, this replacement may appear problematic, as such an outcome may condition a feedforward Clifford operator. For concreteness, consider a scenario where the $T$ gate teleportation protocol is followed by the $X$ measurement; see Fig.~\ref{fig:MGT_AFT}(a). If the actual outcome of the first measurement and the initial inferred value disagree, then the second measurement is performed in the wrong basis. Consequently, one may conclude that classical postprocessing is insufficient to account for such a discrepancy. Another way to formulate this apparent issue is that a general Clifford correction cannot be absorbed into the Pauli frame; thus, the Clifford error arising from the discrepancy between the actual outcome and the initial inferred value seems impossible to correct later solely by classical postprocessing.

We can provide a clear and intuitive resolution of this puzzle within the paradigm of algorithmic fault tolerance. Namely, whenever a logical measurement outcome is 50/50 random and independent, the feedforward operator can be realized by an appropriate Pauli operator, rather than by the $S$ gate; see Fig.~\ref{fig:MGT_AFT}(b). Consequently, even if the logical measurement outcome is incorrectly inferred at that time, then the resulting discrepancy leads to only a Pauli error, which can be accounted for within the Pauli frame and corrected later by classical postprocessing. To justify the existence of the corrective Pauli operator, we invoke Theorem~\ref{th:feedforward_reduction}. Concretely, express the Pauli observable for each measurement by propagating it backward through the Clifford circuit to the initial state. If the considered measurement outcome is 50/50 random and independent of all earlier outcomes, then there exists a Pauli stabilizer of the initial state that anticommutes with the resulting observable for this measurement and commutes with those for all earlier measurements~\cite{zhou_algorithmic_2024, cain2025fast}. Equivalently, if we propagate the stabilizers of the initial state forward, then we obtain a Pauli operator $G$ that stabilizes the input state of the $T$ gate teleportation protocol and anticommutes with $P'=Z$. Theorem~\ref{th:feedforward_reduction} therefore applies, allowing us to replace the Clifford feedforward $S^m$ with the Pauli feedforward $G^m$.

From a practical standpoint, our observation allows us to reduce the number of feedforward Clifford operators. For example, consider the 15-to-1 distillation protocol; see Fig.~\ref{fig:tft_pauli_feedforward}(a). Immediately before the layer of $T$ gates, the circuit prepares a stabilizer state $\rho_{\mathrm{in}}$ whose stabilizer group includes five independent $X$-type stabilizers specified in Fig.~\ref{fig:tft_pauli_feedforward}(b). Using Theorem~\ref{th:feedforward_reduction}, these stabilizers allow five of the fifteen feedforward $S$ gates to be replaced by Pauli operators. Thus, the number of teleportation locations potentially requiring a feedforward Clifford operator is reduced from fifteen to ten, with the five feedforward operators handled through Pauli-frame updates. For instance, the five $X$-type stabilizers can replace the feedforward for the $T$ gates on qubits 2, 4, 3, 5 and 9, respectively; another possibility of replacing the feedforward would be for qubits 3, 2, 8, 15 and 13, justifying the observation made in the section ``State Distillation Factories'' of Ref.~\cite{zhou_algorithmic_2024}. We remark that replacing selected feedforward Clifford operators by Pauli operators may change the effective measurement outcomes controlling the remaining feedforward operations, which can be tracked through Pauli-frame updates.

\subsection{Ancilla-assisted Pauli feedforward}

\begin{figure*}[ht!]
    \centering
    \includegraphics[width=.92\linewidth]{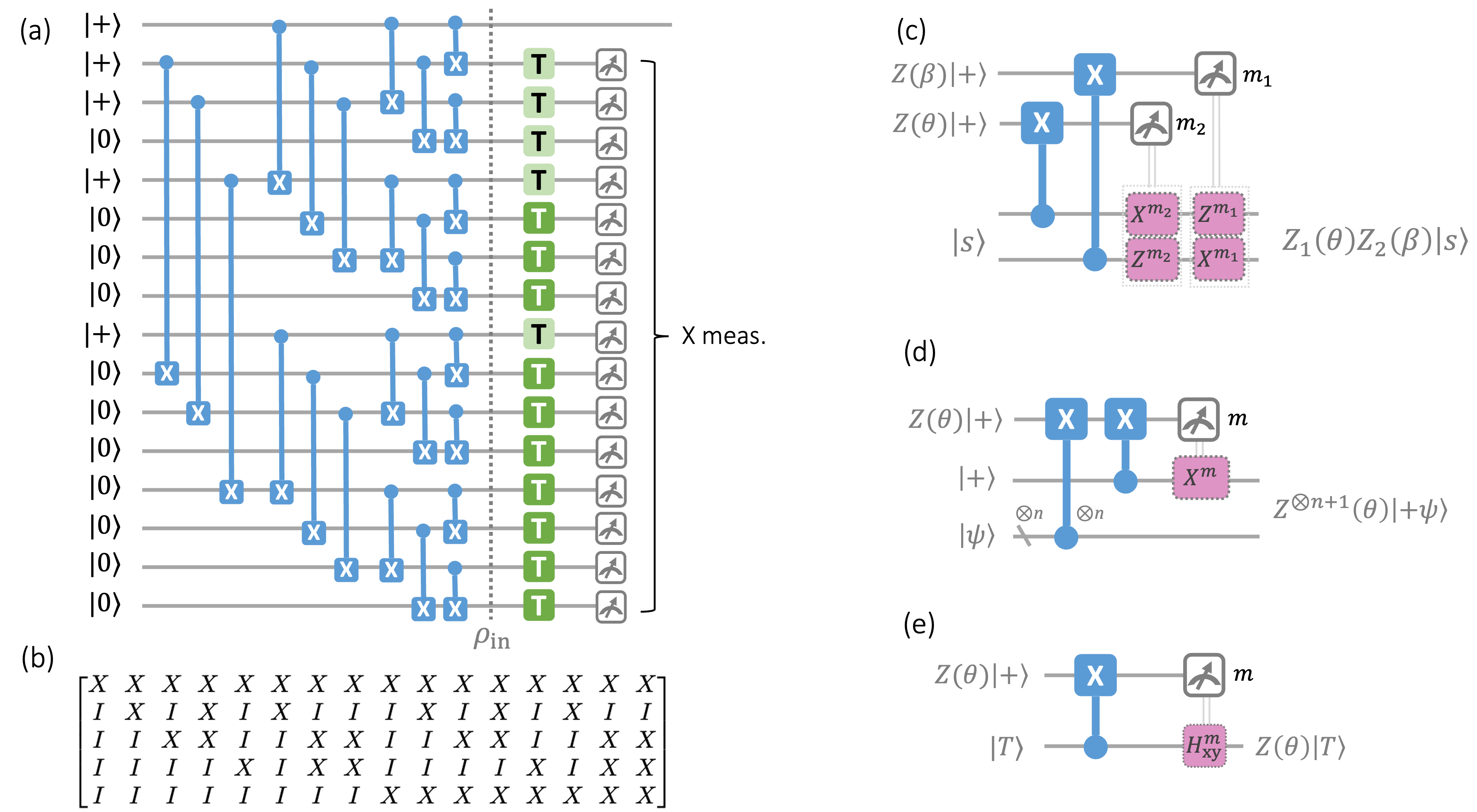}
\caption{ (a) The 15-to-1 distillation circuit (adopted from Ref.~\cite{beverland_cost_2021}), with the stabilizer state $\rho_{\mathrm{in}}$ immediately before the layer of $T$ gates. Qubits are labeled 1 through 16 from the top to the bottom. Five T gates indicated with black text can be implemented using deterministic MGT with Pauli feedforward operators. In (b), we specify five $X$-type stabilizers of $\rho_{\mathrm{in}}$. (c) A two-qubit example illustrating a ``simple'' feedforward operator, i.e., for a suitably stabilized input state $\ket{s}$ a feedforward operator in the sequential teleportation of $Z_1(\theta)$ and $Z_2(\beta)$ can be implemented by an appropriate Pauli operator, yielding $Z_1(\theta)Z_2(\beta)\ket{s}$. (d) Ancilla-assisted MGT for an arbitrary $n$-qubit input state $\ket{\psi}$, where one ancilla in $\ket{+}$ allows one to replace the non-Clifford feedforward operator by a Pauli $X$ operator on the ancilla. The output is $Z^{\otimes n+1}(\theta)\ket{+\psi}$, or more generally $(Z\otimes \widetilde P)(\theta)\ket{+\psi}$ for arbitrary Pauli $\widetilde P$. (e) One can replace the usual non-Clifford feedforward $Z(2\theta)$ operator by the Clifford operator $H_{\mathrm{xy}}=(X+Y)/\sqrt{2}$ when teleporting $Z(\theta)$ on $\ket{T}$.}
    \label{fig:tft_pauli_feedforward}
\end{figure*}

An interesting extension of Theorem~\ref{th:feedforward_reduction} is that one can remove the dependence on the input state by introducing a stabilizer ancilla state. The following corollary allows us to replace the feedforward operator in a deterministic MGT protocol with a Pauli feedforward operator on the ancilla system.

\begin{corollary}
Consider the resource state $\ket{\eta}$ described in Theorem~\ref{th:1}. Let $\ket{\widetilde{s}}$ be a stabilizer state with a set of independent stabilizer generators $\{G_k\}$. For an arbitrary set of commuting Pauli operators $\{Q_k\}$, such that for all $k$ and $\ell \neq k$
\begin{equation}
    \{G_k,Q_k\}=0
    \quad \text{and} \quad
    [G_k,Q_\ell]=0,
\end{equation}
there exists a deterministic MGT protocol that uses $\ket{\eta}$ and implements
\begin{equation}
    \widetilde U_{\bm 0}=\prod_j \widetilde P_j(\theta_j),
\end{equation}
on the state $\ketbra{\widetilde s}{\widetilde s}\otimes\rho_{\mathrm{in}}$, where $\widetilde P'_k=Q_k\otimes P'_k$ and $\widetilde P_j=\prod_{k\in\mathcal I_j}\widetilde P'_k$, with the following Pauli feedforward operators
\begin{equation}
    F_{\bm m}=\prod_{k} G_k^{m_k} \otimes I
\end{equation}
for the outcome $\bm m = \{m_k\}$.
\label{cor:ancilla-pauli-feedforward}
\end{corollary}

\begin{proof}
Note that the input state for such MGT protocol is $\widetilde\rho_{\mathrm{in}}=\ketbra{\widetilde s}{\widetilde s}\otimes\rho_{\mathrm{in}}$. Since every $G_k$ stabilizes $\ket{\widetilde s}$, then $G_k \otimes I$ stabilizes $\widetilde\rho_{\mathrm{in}}$. Moreover, because $\widetilde P'_k=Q_k\otimes P'_k$, for all $k$ and $\ell\neq k$ we have
\begin{equation}
\{G_k\otimes I,\widetilde P'_k\}=0
\quad\text{and}\quad
[G_k\otimes I,\widetilde P'_\ell]=0.
\end{equation}
Therefore, the MGT protocol satisfies the conditions of Theorem~\ref{th:feedforward_reduction} with input state $\ketbra{\widetilde s}{\widetilde s}\otimes\rho_{\mathrm{in}}$, and the claim about Pauli feedforward operators follows immediately.
\end{proof}

Below we provide an example for the single-generator case of Corollary~\ref{cor:ancilla-pauli-feedforward}.
\begin{example}
 Let the ancilla state be $\ket{\widetilde s}=\ket{+}$. For an arbitrary $n$-qubit input state $\rho_{\mathrm{in}}$ and a resource state $Z(\theta)\ket{+}$, there is a deterministic MGT protocol for  $(Z\otimes Z^{\otimes n})(\theta)$ on $\ketbra{+}{+}\otimes\rho_{\mathrm{in}}$ with Pauli feedforward $F_m=X^m\otimes I$; see Fig.~\ref{fig:tft_pauli_feedforward}(d).
\end{example}
Usually, a deterministic MGT protocol that implements $Z^{\otimes n}(\theta)$ would require a feedforward operator $Z^{\otimes n}(2\theta)^m$, which would be non-Clifford. The ancilla-assisted protocol allows us to replace this feedforward operator with the Pauli operator $X^m\otimes I$, at the price of implementing the gate $(Z\otimes Z^{\otimes n})(\theta)$ on $\ketbra{+}{+}\otimes\rho_{\mathrm{in}}$ rather than $Z^{\otimes n}(\theta)$ on $\rho_{\mathrm{in}}$. This tradeoff can be useful for logical state preparation, for example in the ansatz preparation~\cite{peruzzo_variational_2014,kandala_hardware-efficient_2017,farhi_qaoa_2014}.

\subsection{Extension to general feedforward operators}

A natural generalization of Theorem~\ref{th:feedforward_reduction} is to ask whether the same mechanism can still work when the feedforward operators are allowed to be general unitaries. This relaxation would enlarge the class of input states for which the feedforward operators can be replaced by potentially simpler feedforward operators. Below we provide such a generalization, which is formally identical to Theorem~\ref{th:feedforward_reduction}, except that the operators that stabilize the input state can be any unitaries instead of Pauli operators.

\begin{corollary}
The condition on $\{G_k\}$ being Pauli operators in Theorem~\ref{th:feedforward_reduction} can be relaxed, i.e., if $\{G_k\}$ are unitary stabilizers of $\rho_{\mathrm{in}}$ satisfying the same commutation conditions with $\{P'_k\}$ in Eq.~\eqref{eq_1}, then the same feedforward replacement $F_{\bm m}=\prod_k G_k^{m_k}$ holds.
\label{cor:unitary-stabilizer-feedforward}
\end{corollary}

We omit the proof here as the proof strategy is literally the same as Theorem~\ref{th:feedforward_reduction}. We then consider the following example.
\begin{example}
Consider a deterministic MGT protocol that uses $Z(\theta)\ket{+}$ and implements $Z(\theta)$ gate. If the input state is $\rho_{\mathrm{in}}=\ketbra{T}{T}$, then the feedforward operator $Z^m(2\theta)$ can be replaced by a Clifford feedforward operator $H^m_{\mathrm{xy}}$, where $H_{\mathrm{xy}}= \frac{X+Y}{\sqrt{2}}$; see Fig.~\ref{fig:tft_pauli_feedforward}(e).
\end{example}
The state $\ket{T}$ is not stabilized by any Pauli operator that anticommutes with $Z$, so Theorem~\ref{th:feedforward_reduction} does not apply. However, $\ket{T}$ is stabilized by the Clifford operator $H_{\mathrm{xy}}=(X+Y)/\sqrt{2}$, which satisfies $H_{\mathrm{xy}}\ket{T}=\ket{T}$ and $\{H_{\mathrm{xy}},Z\}=0$. Therefore, in a deterministic MGT protocol that implements $Z((-1)^m\theta)$ conditioned on the outcome $m$, applying $H^m_{\mathrm{xy}}$ as the feedforward operator results in a deterministic MGT protocol for $Z(\theta)$ on input $\rho_{\mathrm{in}}=\ketbra{T}{T}$; see Fig.~\ref{fig:tft_pauli_feedforward}(e).

\section{Discussion}

In this work, we analyzed the structure of and the constraints on MGT. We showed that not every nonstabilizer state is directly useful for MGT if the protocol must reveal no information about the input state, a requirement that is particularly relevant in the context of long quantum computations. We also identified sufficient conditions under which the feedforward operator in an MGT protocol can be implemented as a Pauli operator. This, in turn, can simplify certain state distillation protocols by allowing the feedforward operators to be absorbed into the Pauli frame update.

Our work can be extended in several directions. First, generalizing Theorem~\ref{th:multiqubit-teleportable} to arbitrary MGT protocols would, together with Theorem~\ref{th:1}, provide necessary and sufficient conditions for multiqubit resource states to be directly useful for MGT. Second, from a practical standpoint, it would be important to design and optimize circuit implementations of MGT that take full advantage of readily available, albeit imperfect, quantum hardware. Third, more exotic resource states naturally give rise to strategies beyond compilation into Clifford and $T$ gates, reminiscent of the STAR architecture~\cite{akahoshi_partially_2024}. It would be interesting to investigate in detail how different gates can be used to reduce compilation overhead, and whether complex resource states can be prepared to implement multiple gates simultaneously, thereby reducing the overall space and time overhead of quantum computation.

In addition, resource states that are not directly useful for MGT may still facilitate MGT if they can be probabilistically converted into directly useful states. The conversion of $\ket{F}$ discussed above provides one example~\cite{bravyi_universal_2005}, and more general magic-state conversions have been studied in Refs.~\cite{ahmadi_quantification_2018,rizzo_universal_2026}. Determining which starting states admit such conversions, along with their success probabilities and resource costs, is an interesting direction for future work.

\begin{acknowledgments}
Y.Z. acknowledges useful discussions with Charles Cao, Mark Howard, Liang Jiang and Adam Paetznick. We thank Aidan Wilber-Gauthier for pointing out an issue in an earlier version of Lemma~\ref{lem:uniform-diagonal-state}. This research was sponsored by the NSF (QLCI, Award No. OMA-2120757) and the U.S. Army Research Office (ARO) under grant W911NF-23-1-0051. The views and conclusions contained in this document are those of the authors and should not be interpreted as representing the official policies, either expressed or implied, of the ARO or the U.S. Government.

\end{acknowledgments}

\appendix
\section{Constructing decoding circuits for stabilizer codes}
\label{app:gottemans_algorithm}

An $[\![n,k]\!]$ stabilizer code can be specified by a symplectic matrix $H' = [H'_{\mathrm{X}}|H'_{\mathrm{Z}}]$, where both $H'_{\mathrm{X}}$ and $H'_{\mathrm{Z}}$ are defined on $\mathbb{F}^{\, (n-k)\times n}_2$. By Gaussian elimination, the parity check matrix can be brought into the standard form
\begin{equation}
    H= [H_{\mathrm{X}}|H_{\mathrm{Z}}] =\begin{bmatrix}[ccc|ccc]
        I & A_1 & A_2 & B & 0 & C \\
        0 & 0 & 0 & D & I & E
    \end{bmatrix},
    \label{eq:standard_form}
\end{equation}
where the two identity blocks have row sizes $r=\mathrm{rank}(H_{\mathrm{X}})$ and $n-k-r$, respectively, and the three column blocks have sizes $r$, $n-k-r$, and $k$. Once the parity-check matrix is in standard form, the logical operators can be read off immediately:
\begin{align}
      &L_{\mathrm{X}} =  \begin{bmatrix}[ccc|ccc]
        0&E^T&I&C^T & 0 & 0
    \end{bmatrix},\\
    &L_{\mathrm{Z}} =     \begin{bmatrix}[ccc|ccc]
        0&0&0&A_2^T & 0 & I
    \end{bmatrix}.
\end{align}
Each row corresponds to one independent logical operator. The encoding circuit $W_{\mathrm{enc}}$ for the standard form of $\mathcal{Q}$ can be efficiently evaluated from the standard form matrices using Gottesman's algorithm~\cite{gottesman_stabilizer_1997}, and the decoding circuit $V$ is then obtained through $V=W^\dagger_{\mathrm{enc}}$. Alternative approaches to synthesizing stabilizer encoding circuits are discussed in Ref.~\cite{wilde_quantum_2008}. We now only consider the construction of the encoding circuit. The following Algorithm~\ref{alg:gottesman} gives an explicit way to construct $W_{\mathrm{enc}}$ for stabilizer code in the standard form.

The logic of Algorithm~\ref{alg:gottesman} is straightforward. $W_{\mathrm{enc}}$ should map single-qubit Pauli $Z$ operators on the first $n-k$ qubits to stabilizer generators defined via $H$, and single-qubit Pauli $X(Z)$ on the last $k$ qubits to logical $X(Z)$ operators defined via $L_X(L_Z)$. The first loop uses the support of the logical $X$ operators $L_{\mathrm{X}}$ to couple each physical qubit corresponding to a logical qubit to the physical qubits on which the corresponding encoded logical $X$ must act. In this way, the single-qubit Pauli $X$ operators are propagated to the correct logical-$X$ operators. The second loop then propagates the stabilizers. Applying a Hadamard gate to qubit $j$ maps its Pauli $Z$ operator to Pauli $X$, after which controlled-$X$, controlled-$Z$, and controlled-$Y$ gates reproduce the non-pivot Pauli factors of the corresponding stabilizer generator. The phase gate $S$ accounts for a Pauli $Y$ on the pivot qubit. If desired, a controlled-$Y$ gate can be decomposed as $\mathrm{CY}(j,\ell)=S_j^\dagger\mathrm{CZ}(j,\ell)\mathrm{CX}(j,\ell)$.

\begin{center}
\begin{minipage}{0.95\columnwidth}
\hrule
\vspace{0.6ex}
\refstepcounter{algorithm}
\noindent {\bf{Algorithm~\thealgorithm.}} Constructing an encoding circuit $W_{\mathrm{enc}}$
\hrule
\label{alg:gottesman}
\vspace{0.4ex}
\begin{algorithmic}[1]
\Require $H_{\mathrm{X}}, H_{\mathrm{Z}}  \in \mathbb{F}_2^{\,(n-k)\times n}$,  $L_{\mathrm{X}} \in \mathbb{F}_2^{\, k\times 2n}$
\Ensure A quantum circuit on $n$ qubits

\For{$j = 1$ to $k$}
    \For{$\ell = 1$ to $n-k$}
        \If{$L_{\mathrm{X}}[j,\ell] = 1$}
            \State Append $\mathrm{CX}(n-k+j,\ell)$
        \EndIf
    \EndFor
\EndFor
\For{$j = 1$ to $r$}
    \State Append Hadamard gate $H$ on qubit $j$
    \For{$\ell = 1$ to $n$}
        \If{$\ell\neq j$}
            \If{$H_{\mathrm{X}}[j,\ell] = 1$ and $H_{\mathrm{Z}}[j,\ell] = 0$}
                \State Append $\mathrm{CX}(j,\ell)$
            \ElsIf{$H_{\mathrm{X}}[j,\ell] = 0$ and $H_{\mathrm{Z}}[j,\ell] = 1$}
                \State Append $\mathrm{CZ}(j,\ell)$
            \ElsIf{$H_{\mathrm{X}}[j,\ell] = 1$ and $H_{\mathrm{Z}}[j,\ell] = 1$}
                \State Append $\mathrm{CY}(j,\ell)$
            \EndIf
        \EndIf
    \EndFor
    \If{$H_{\mathrm{Z}}[j,j] = 1$}
        \State Append phase gate $S$ on qubit $j$
    \EndIf
\EndFor
\end{algorithmic}
\vspace{0.6ex}
\hrule
\end{minipage}
\end{center}

Once the encoding circuit $W_{\mathrm{enc}}$ for the standard-form representation $H$ has been obtained, an encoding circuit $W'_{\mathrm{enc}}$ for the original representation $H'$ can be constructed by undoing the Gaussian-elimination procedure. The encoding circuit constructed here depends on the chosen ordered set of stabilizer generators. Adding row $j$ to row $k$ is implemented by a $\operatorname{CX}$ gate controlled on ancilla $j$ and targeted on ancilla $k$. Swapping physical-qubit columns corresponds to a $\operatorname{SWAP}$ gate on the encoded output qubits. Since these operations and $W_{\mathrm{enc}}$ are Clifford, the column operations can be propagated through $W_{\mathrm{enc}}$ and combined with the row operations into a single Clifford circuit $W_{\mathrm{ge}}$. Thus, the encoder for the original representation can be written as
\begin{equation}
    W'_{\mathrm{enc}} = W_{\mathrm{enc}} W^\dagger_{\mathrm{ge}},
\end{equation}
and the decoding circuit associated with $H'$ is then given by $V = W'_{\mathrm{enc}}{}^\dagger = W_{\mathrm{ge}}W_{\mathrm{enc}}^\dagger$.

\bibliographystyle{apsrev4-2}
\bibliography{ref.bib}
\end{document}